\documentclass[10pt, conference, letterpaper]{IEEEtran}
\usepackage{cite}
\usepackage{amsmath,amssymb,amsfonts}
\usepackage{amsthm}
\usepackage{graphicx}
\usepackage{textcomp}
\usepackage{xcolor}
\usepackage{amsfonts}
\usepackage{amssymb}
\usepackage{amsthm}
\usepackage{cases}
\usepackage{subfigure}
\usepackage{booktabs}
\usepackage{colortbl}
\usepackage{tabularx}
\usepackage{algorithm}
\usepackage{algorithmic}
\usepackage{color}
\usepackage{subfigure}
\usepackage{url}
\usepackage{setspace}
\newtheorem{theorem}{Theorem}%[section]
\newtheorem{lemma}[theorem]{Lemma}

\newtheorem{corollary}[theorem]{Corollary}

\begin{document}

% \title{Sparse Gaussian Compressive Sensing Measurement Matrices with Applications to Lossy Transmission in IoT
% % Sparse Gaussian Measurement Matrix Design and Its Application in Lossy Transmission in IoT
% % {\footnotesize \textsuperscript{*}Note: Sub-titles are not captured in Xplore and
% % should not be used}
% % \thanks{Identify applicable funding agency here. If none, delete this.}
% }
% \title{On Designing Projection Matrices to Achieve Efficient Compressive Data Collection in IoT}
\title{Towards Efficient Compressive Data Collection \\in the Internet of Things}
% Sparse Gaussian Measurement Matrix Design and Its Application in Lossy Transmission in IoT
% {\footnotesize \textsuperscript{*}Note: Sub-titles are not captured in Xplore and
% should not be used}
% \thanks{Identify applicable funding agency here. If none, delete this.}

\author{\IEEEauthorblockN{Peng Sun$^{\dagger}$, Liantao Wu$^{\wr}$, and Zhi Wang$^{\S}$}
% \IEEEauthorblockA{$^{\dagger}$State Key Laboratory of Industrial Control Technology, Zhejiang University, P. R. China}
\IEEEauthorblockA{$^{\dagger}$School of Science and Engineering, The Chinese University of Hong Kong, Shenzhen, P. R. China}
% \IEEEauthorblockA{$^{\ast}$Institute of Cyber Science and Technology, Zhejiang University, P. R. China}

\IEEEauthorblockA{$^{\wr}$Shanghai Institute of Fog Computing Technology (SHIFT), ShanghaiTech University, P. R. China}
\IEEEauthorblockA{$^{\S}$State Key Laboratory of Industrial Control Technology, Zhejiang University, P. R. China}
% \IEEEauthorblockA{$^{\wr}$Department of Computer Science, University of Illinois at Urbana-Champaign, USA}

% \IEEEauthorblockA{$^{\P}$Key Laboratory of Intelligent Computing \& Information Processing of Ministry of Education,\\
% College of Information Engineering, Xiangtan University, P. R. China}

Email: sunpeng@cuhk.edu.cn, wult@shanghaitech.edu.cn, zjuwangzhi@zju.edu.cn}

\markboth{Journal of \LaTeX\ Class Files,~Vol.~14, No.~8, August~2015}%
{Shell \MakeLowercase{\textit{et al.}}: Bare Demo of IEEEtran.cls for IEEE Journals}

\maketitle
% \footnotetext[1]{Liantao Wu is the corresponding author.}
\begin{abstract}
It is of paramount importance to achieve efficient data collection in the Internet of Things (IoT). Due to the inherent structural properties (e.g., sparsity) existing in many signals of interest, compressive sensing (CS) technology has been extensively used for data collection in IoT to improve both accuracy and energy efficiency. Apart from the existing works which leverage CS as a channel coding scheme to deal with data loss during transmission, some recent results have started to employ CS as a source coding strategy. The frequently used projection matrices in these CS-based source coding schemes include dense random matrices (e.g., Gaussian matrices or Bernoulli matrices) and structured matrices (e.g., Toeplitz matrices). However, these matrices are either difficult to be implemented on resource-constrained IoT sensor nodes or have limited applicability. To address these issues, in this paper, we design a novel simple and efficient projection matrix, named sparse Gaussian matrix, which is easy and resource-saving to be implemented in practical IoT applications. We conduct both theoretical analysis and experimental evaluation of the designed sparse Gaussian matrix. The results demonstrate that employing the designed projection matrix to perform CS-based source coding could significantly save time and memory cost while ensuring satisfactory signal recovery performance.

% There exist two main problems in currently existing dense measurement matrices for compressed sensing (CS) applications in the wireless sensor netowrks (WSNs) and internet of things (IoT). One is the difficulty of hardware implementation due to limited computation capability and memory. The other is low sensing/transmission efficiency. In this paper, we present a novel simple and efficient measurement matrix, sparse Gaussian matrix. The proposed measurement matrix has the following advantages (1) It has easy hardware implementation because of
% the matrix structure; (2) It has high sensing efficiency because of the highly sparse structure; (3) It is incoherent with different popular
% sparsity basis¡¯ like wavelet basis and fourier basis; Through theoretical analysis, the RIP of sparse Gaussian matrix is proved and a block diagonal version of sparse Gaussian is also proposed. An immediate application is in sparse signal transmission via lossy link, where experimental results show the new sparse Gaussian matrix needs much less construction cost while guaranteeing almost the same performance compared with dense Gaussian or bernoulli matrix. The block diagonal sparse Gaussian matrix can further reduce the latency induced at the sensor node.
\end{abstract}


\IEEEpeerreviewmaketitle

\section{Introduction}

\IEEEPARstart{W}{ith} the prevalence of smart devices, sensors, RFID tags, etc., and the broad deployment of wireless communication infrastructures, the Internet of Things (IoT) is undergoing an unprecedentedly rapid development~\cite{qiu2018can,wu2020toward}. Recent years have witnessed a wide range of IoT applications, such as environmental monitoring, intelligent transportation systems, smart cities, and smart homes~\cite{qian2019internet}, to name a few.  

Data collection is the fundamental operation in IoT to underpin various applications. The IoT sensor nodes (e.g., smart devices, sensors) need to periodically report their sensed data to a fusion center (FC), where information aggregation and extraction tasks are performed. Generally, IoT sensor nodes are resource-constrained. More precisely, they have limited computational capability, storage capacity, and power supply. Therefore, it is of crucial importance to achieve efficient data collection in IoT. Many real-world signals of interest (e.g., temperature, noise level, air quality level, and traffic speed) present inherent structural properties (e.g., sparsity). Thus, compressive sensing (CS)~\cite{baraniuk2007compressive} has been widely used to achieve efficient information acquisition and transmission in IoT, as it promises accurate recovery of high-dimensional sparse signals from only a small number of random measurements.

% Thus far, a great amount of research efforts have been devoted to utilizing CS for compressive data collection in IoT~\cite{sun2018prss,wu2015efficient,wu2017sparse,kong2013data}. Most of the existing works leverage CS as a channel coding scheme to deal with data loss during information collection in IoT. For instance, the work in~\cite{kong2013data} mined and analyzed four different data loss patterns and proposed an environmental space time improved compressive sensing (ESTICS) algorithm to accurately reconstruct the massive missing data. Our previous work~\cite{wu2015efficient} modeled the random data loss during transmission over a lossy wireless link as a CS measurement process, where the corresponding projection matrix is determined by the correctly received data samples. Some recent works~\cite{li2017unbalanced} have started to employ CS as a source coding strategy to further improve accuracy and energy-efficiency of information collection. Therein, a simple and straightforward CS-based source coding scheme is to perform random projection on the original sensed signals at the IoT sensor nodes before transmission, where some proper projection matrices are utilized. 
Thus far, a significant amount of research efforts have been devoted to utilizing CS for compressive data collection in IoT~\cite{sun2018prss,wu2015efficient,wu2017sparse,kong2013data}. Most of the existing works leverage CS as a channel coding scheme to deal with data loss during information collection in IoT. For instance, the work in~\cite{kong2013data} mined and analyzed four different data loss patterns and proposed an environmental space-time improved compressive sensing (ESTICS) algorithm to reconstruct the massive missing data accurately. Our previous work~\cite{wu2015efficient} modeled the random data loss during transmission over a lossy wireless link as a CS measurement process, where the correctly received data samples determine the corresponding projection matrix. Some recent works (e.g.,~\cite{li2017unbalanced}) have started to employ CS as a source coding strategy to further improve the accuracy and energy efficiency of information collection. A simple and straightforward CS-based source coding scheme is to perform random projection on the original sensed signals at the IoT sensor nodes before transmission, where some proper projection matrices are utilized.

In CS-based source coding schemes, the projection matrix is a critical factor for achieving satisfactory signal reconstruction performance. Specifically, the CS theory establishes that the projection matrix needs to be incoherent with the transform basis, where the original signal has a sparse representation, to guarantee accurate signal recovery. Frequently used projection matrices in CS can be divided into three categories. The first one is standard dense random matrices, e.g., random Gaussian or Bernoulli matrices. Elements in these matrices are independently and identically distributed random variables with certain distribution functions, i.e., Gaussian distribution or Bernoulli distribution. Existing works have demonstrated that Dense random matrices are incoherent with most sparsifying bases and satisfy the restricted isometry property (RIP), which promises accurate signal recovery with high probability. However, as stated in~\cite{mamaghanian2011compressed}, the implementation of dense random Gaussian/Bernoulli projection matrices requires the integration of a Gaussian/Bernoulli-distributed random number generator and the computation of complicated matrix multiplication. This is too complex and resource-consuming for low-cost IoT sensor nodes with limited computational capability and storage capacity. The second kind of projection matrices is structured projection matrices, e.g., Toeplitz matrices. Elements of this kind of matrices are arranged with a particular pattern, yielding a lower degree of randomness than dense random matrices. The third kind of projection matrices is partial basis matrices, such as partial Fourier matrices, constructed by selecting a subset of rows from an orthogonal Fourier transform basis. The essential requirement of the projection matrix used in CS-based source coding is that given a certain number of data samples received at the FC, we could consistently achieve accurate signal recovery regardless of the transform basis under which the target signal can be sparsely represented. However, as stated in~\cite{zhang2010compressed}, partial Fourier matrices can only adapt to the time domain sparse signals, while as indicated in~\cite{bajwa2007toeplitz}, Toeplitz matrices are neither applicable to all sparse signals. Therefore, the shortcoming of structured projection matrices and partial basis matrices lies in that they are confined to specific sparse signals, leading to a limited range of applications. To summarize, none of the existing projection matrices can satisfy our requirement of efficient and wide-area applications of CS techniques in IoT data collection.

% To address the above-mentioned issues, in this paper, we design a simple and efficient projection matrix, named sparse Gaussian matrix, which is easy and efficient to be implemented on low-cost IoT sensor nodes and is applicable to various kinds of sparse signals. Specifically, we first present the construction of a sparse Gaussian matrix, and then theoretically prove its RIP, which guarantees accurate signal recovery with high probability. Finally, we conduct extensive experiments, and the results demonstrate that employing the designed sparse Gaussian matrix to perform CS-based source coding can significantly save time and memory cost for IoT sensor nodes while guaranteeing satisfactory signal recovery performance at the FC compared to existing projection matrices. 
To address the issues mentioned above, in this paper, we design a simple and efficient projection matrix, named sparse Gaussian matrix, which is easy and efficient to be implemented on low-cost IoT sensor nodes and applies to various kinds of sparse signals. Specifically, we first present the construction of a sparse Gaussian matrix and then theoretically prove its RIP, which guarantees accurate signal recovery with high probability. Finally, we conduct extensive experiments to evaluate the performance of the designed sparse Gaussian matrix. The results demonstrate that employing the designed sparse Gaussian matrix to perform CS-based source coding can significantly save time and memory cost for IoT sensor nodes while guaranteeing satisfactory signal recovery performance at the FC compared to existing projection matrices.

In summary, this paper makes the following contributions: 

% \begin{itemize}

% \item To the best of our knowledge, this is the first work to consider how to design a proper projection matrix to realize efficient and resource-saving CS-based source coding for efficient data collection in IoT.  
% \item We design a novel simple projection matrix, named sparse Gaussian matrix, and theoretically prove its RIP. 
% \item We conduct extensive experiments. The results demonstrate that leveraging the designed sparse Gaussian matrix to perform CS-based source coding can significantly reduce the time and memory cost for the IoT sensor nodes while guaranteeing satisfactory signal recovery performance at the FC. 
% \end{itemize}
\begin{itemize}

\item To the best of our knowledge, this is the first work to consider designing a proper projection matrix to realize efficient and resource-saving CS-based source coding for efficient data collection in IoT.  
\item We design a novel, simple projection matrix, named sparse Gaussian matrix, and theoretically prove its RIP. 
\item We conduct extensive experiments and the results demonstrate that leveraging the designed sparse Gaussian matrix to perform CS-based source coding can significantly reduce the time and memory cost for the IoT sensor nodes while guaranteeing satisfactory signal recovery performance at the FC. 
\end{itemize}

\section{System Model and Problem Formulation}

\label{sec5}
In this section, we present the system model and problem formulation.

\begin{figure}
\centering
\includegraphics[width=0.8\columnwidth]{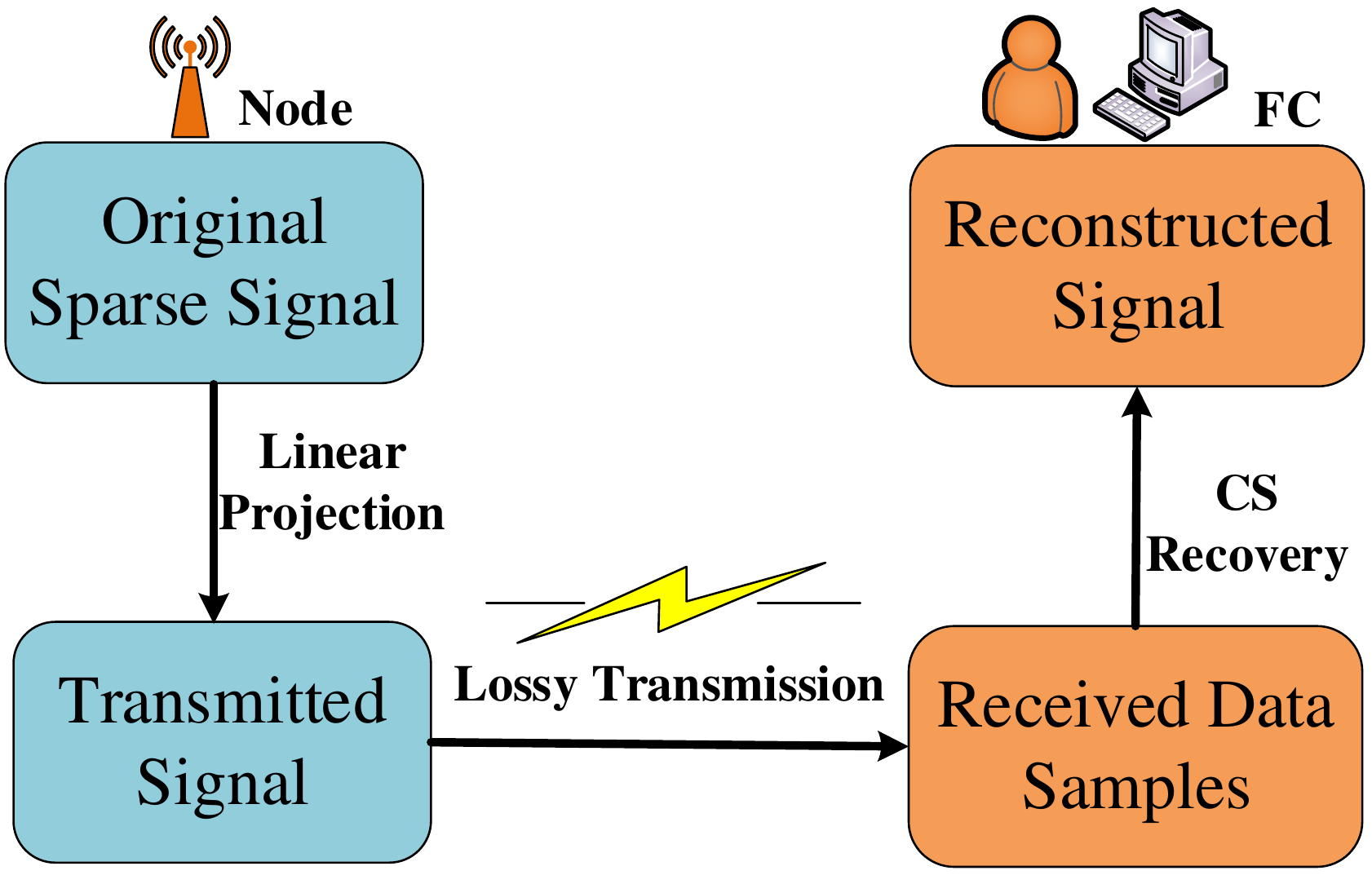}
\caption{System model. }
\label{fig:csLink}
% \vspace{-20pt}
\end{figure}

\subsection{System Model}
% In this paper, we consider a typical data collection process in IoT, which is fundamental to enable various applications. As shown in Fig.~\ref{fig:csLink}, the whole process can be described as follows. The IoT sensor node periodically senses the physical world, and in a specific sensing period, the node will collect a certain number of data samples and generate an original sensed signal which is sparse under some transform basis (e.g., Fourier basis). The data samples are then transmitted to the FC for further information aggregation and extraction. Due to factors like channel noise, congestion, packet collisions, and limited transmission power, data transmission in IoT is prone to data loss, that is, only a subset of transmitted data samples can successfully arrive at the FC. To deal with the data loss, based on the sparse structure of the signal, our previous work~\cite{wu2015efficient} proposed to model the random data loss during transmission as a CS measurement process. The corresponding projection matrix modeling data loss is a sparse binary measurement matrix which contains one and only one ``$1$" element in each row and at most one ``$1$" in each column, and ``$0$" everywhere else (Note that the formal definition of this projection matrix will be presented in the next subsection). It has been shown in~\cite{wu2015efficient} that CS-based sparse signal transmission method achieves much better transmission performance compared to traditional retransmission-interpolation schemes. 
This paper considers a typical data collection process in IoT, which is fundamental to enable various applications. Fig.~\ref{fig:csLink} shows the whole process of data collection. Specifically, the IoT sensor node periodically senses the physical world. In a specific sensing period, the node will collect a certain number of data samples and generate an original sensed signal with a sparse representation under some transform basis (e.g., Fourier basis). The data samples are then transmitted to the FC for further information aggregation and extraction. Due to channel noise, congestion, packet collisions, and limited transmission power, data transmission in IoT is prone to data loss. That is, only a subset of transmitted data samples can successfully arrive at the FC. To deal with the data loss, based on the sparse structure of the signal, our previous work~\cite{wu2015efficient} proposed to model the random data loss during transmission as a CS measurement process. The corresponding projection matrix modeling data loss is a sparse binary measurement matrix which contains one and only one ``$1$" element in each row and at most one ``$1$" in each column, and ``$0$" everywhere else (Note that the formal definition of this projection matrix will be presented in the next subsection). It has been shown in~\cite{wu2015efficient} that the CS-based sparse signal transmission method achieves much better transmission performance compared to traditional retransmission-interpolation schemes.

% Moreover, as pointed out in our previous work~\cite{wu2017sparse}, for a large number of applications in IoT, the transmission rate is much less than the sampling rate. For example, the upper bound of the transmission rate in IEEE 802.15.4 is 250kbps, while the sampling rate of commercial AD converters is several MHz. Moreover, when multiple sensor nodes transmit data to the fusion center (FC) via random access, where packet collisions may happen, a larger transmission rate generally leads to more severe packet collisions, worsening the signal recovery performance at the FC. Thus, we introduce a simple CS-based source coding operation at the sensor node before transmission. Specifically, instead of directly transmitting the original sensed signal which is composed of all the acquired data samples, it will be projected onto a lower-dimensional space before being transmitted. In this way, the number of actually transmitted data samples is smaller than the number of acquired data samples, thus avoiding unnecessary resource cost for IoT sensor nodes. In order to optimize the resource cost of the source coding process for the IoT sensor nodes and ensure satisfactory signal recovery performance at the FC, the projection matrix used in CS-based source coding should be carefully designed. 
Moreover, as pointed out in our previous work~\cite{wu2017sparse}, for a large number of applications in IoT, the transmission rate is much less than the sampling rate. For example, the upper bound of the transmission rate in IEEE 802.15.4 is 250kbps, while the sampling rate of commercial AD converters is several MHz. Moreover, when multiple sensor nodes transmit data to the fusion center (FC) via random access, where packet collisions may happen, a higher transmission rate generally leads to more severe packet collisions, worsening the signal recovery performance at the FC. Thus, we introduce a simple CS-based source coding operation at the sensor node before transmission. Specifically, instead of directly transmitting the original sensed signal, which comprises all the acquired data samples, it will be projected onto a lower-dimensional space before being sent. In this way, the number of actually transmitted data samples is smaller than the number of acquired data samples, thus avoiding unnecessary resource costs for IoT sensor nodes. To optimize the resource cost of the source coding process for the IoT sensor nodes and ensure satisfactory signal recovery performance at the FC, the projection matrix used in CS-based source coding should be carefully designed. 
% Moreover, it is worth noting that the transmission rate is much lower than the sampling rate in most practical IoT applications~\cite{wu2017sparse}. Thus, there is no need to straightforwardly transmit all the generated data samples to the FC, which leads to unnecessary resource cost for IoT sensor nodes and even worse transmission performance (e.g., more severe packet collisions). Instead, we could perform a CS-based source coding before transmission. In other words, the original signal will be projected onto a lower-dimensional space before being transmitted. In order to optimize the resource cost of the source coding process for the IoT sensor nodes and ensure satisfactory signal recovery performance at the FC, the projection matrix used in source coding should be carefully designed. 

% \vspace{-3pt}

\subsection{Problem Formulation}
In a specific sensing period, an IoT sensor node is assumed to have acquired $N$ data samples from the physical world. Then, we denote the original signal assembled by the $N$ data samples by $\bf{x}$ and ${\bf{x}} = [x_1,x_2,...,x_N]^T$, where $x_i$ ($i = 1,2,...,N$) is the $i$-th data sample. It has been shown in~\cite{wu2017sparse,sun2018prss,sun2019scra} that most target signals in IoT applications have a sparse (compressible) representation in some transform basis, e.g., Fourier, DCT, wavelets, etc. Formally, if we denote by ${\bf{\Psi }}$ the transform basis under which $\bf{x}$ can be sparsely represented, we have ${\bf{x}} = {\bf{\Psi s}}$, and ${\bf{s}}$ is a sparse vector, which contains only a small number of nonzero elements.

As stated above, instead of directly transmitting the original signal $\bf{x}$, it is more desirable to perform a CS-based source coding before transmission, i.e., conducting a random projection on $\bf{x}$ using a well-designed projection matrix ${{\bf{\Phi}}_s} \in {\mathbb{R}}^{{M_s} \times N}$. Formally, this random projection process can be represented as
\begin{equation}
\label{sam}
{{\bf{x}}_s} = {{\bf{\Phi}}_s}{\bf{x}}, 
\end{equation}
where ${{\bf{x}}_s} \in {\mathbb{R}}^{M_s}$ is the signal after the projection operation, and it is the signal actually to be transmitted. ${{\bf{\Phi}}_s}$ is the employed projection matrix.

As previously mentioned, data transmission from the IoT sensor node to the FC is subject to data loss. In the same way as in~\cite{wu2015efficient}, we model the random data loss during transmission as a CS measurement process, which can be represented as
\begin{equation}
\label{sam2}
{\bf{y}} = {{\bf{\Phi}}_r}{{\bf{x}}_s} = {{\bf{\Phi}}_r}{{\bf{\Phi}}_s}{\bf{x}} = {\bf{\Phi}}{\bf{x}} = {\bf{\Phi}}{\bf{\Psi}}{\bf{s}} = {\bf{A}}{\bf{s}},
\end{equation}
where ${{\bf{\Phi}}_r} \in {\mathbb{R}}^{M \times {M_s}}$ is the projection matrix modeling the random data loss, and it is determined by the correctly received data samples. ${\bf{\Phi}} = {{\bf{\Phi}}_r}{{\bf{\Phi}}_s}$ is the equivalent projection matrix, ${\bf{A}} = {\bf{\Phi}}{\bf{\Psi}} = {{\bf{\Phi}}_r}{{\bf{\Phi}}_s}{\bf{\Psi}}$ is called the equivalent sensing matrix, and $\bf{y} \in {\mathbb{R}}^M$ is the measurement vector composed of the successfully received data samples at the FC.

Note that ${{\bf{\Phi}}_r}$ is a sparse binary measurement matrix which contains one and only one ``$1$" element in each row and at most one ``$1$" in each column, and ``$0$" everywhere else. Mathematically, ${{\bf{\Phi}}_r}$ is constructed as
\begin{equation}
{\bf{\Phi }}_r(i,j) = \left\{ {\begin{array}{*{20}{c}}
{1\;\;\;\;\;\;\;\;\;\;\;j{\rm{ = J}}(i)}\\
{0\;\;\;\;\;\;\;\;\text{otherwise}}
\end{array}} \right.,
\label{Phi}
\end{equation}
where $i$ is the row index of ${\bf{\Phi}}_r$, corresponding to the sequence number of the received data samples, and $j$ is the column index. $J(i)$ is the sequence number of the data sample sent by the IoT sensor node while being the $i$-th received data sample at the FC. For example, assume that $7$ data samples are sent from the node, and only $4$ of them are successfully received at the FC. If their corresponding sending sequence number is $1$, $2$, $5$, $7$, then ${\bf{\Phi}}_r$ is constructed as 
\begin{equation}
{{\bf{\Phi}}_r} = \left[ {\begin{array}{*{20}{c}}
1&0&0&{\begin{array}{*{20}{c}}
0&0&0&0
\end{array}}\\
0&1&0&{\begin{array}{*{20}{c}}
0&0&0&0
\end{array}}\\
0&0&0&{\begin{array}{*{20}{c}}
0&1&0&0
\end{array}}\\
0&0&0&{\begin{array}{*{20}{c}}
0&0&0&1
\end{array}}
\end{array}} \right].
\label{randommatrix}
\end{equation}

Finally, at the FC, the original signal $\bf{x}$ can be exactly recovered from the measurement vector $\bf{y}$ which is composed of a sufficient number of successfully received data samples via ${\hat{\bf{x}}} = {\bf{\Psi}}{\hat{\bf{s}}}$, where ${\hat{\bf{s}}}$ is the solution to the following constrained optimization problem
\begin{equation}
{\hat{\bf{s}}} = \arg\min\limits_{\bf{s}}\|{\bf{s}}\|_0,~~\text{s.t.}~{\bf{y}} = {\bf{A}}{\bf{s}},
\label{rec1}
\end{equation}
which can be solved using the orthogonal matching pursuit (OMP) method~\cite{tropp2004greed}. Furthermore, ${\hat{\bf{s}}}$ can also be found from the solution to the following convex optimization problem
\begin{equation}
{\hat{\bf{s}}} = \arg\min\limits_{\bf{s}}\|{\bf{s}}\|_1,~~\text{s.t.}~{\bf{y}} = {\bf{A}}{\bf{s}},
\label{rec2}
\end{equation}
which can be solved efficiently using the basis pursuit (BP) algorithm~\cite{candes2008introduction}. 

Herein, the signal recovery performance can be reflected by signal recovery error, which is calculated as $\frac{{\left\| {{\bf{x}} - {\bf{\hat x}}} \right\|_2}}{{\left\| {\bf{x}} \right\|_2}}$. Obviously, a smaller signal recovery error indicates better signal recovery performance.

Notably, according to CS theory, the equivalent projection matrix ${\bf{\Phi}} = {{\bf{\Phi}}_r}{{\bf{\Phi}}_s}$ plays a critical role in achieving accurate signal recovery. However, the projection matrix ${{\bf{\Phi}}_r}$ which models data loss during transmission is determined by the pragmatic link quality. Specifically, ${{\bf{\Phi}}_r}$ is completely determined by which transmitted data samples are correctly received at the FC. Thus, we cannot design ${{\bf{\Phi}}_r}$. In contrast, we can carefully choose the projection matrix ${{\bf{\Phi}}_s}$ used in CS-based source coding at the IoT sensor node. Traditional projection matrices, including dense random matrices, structured projection matrices, and partial basis matrices, are either difficult to be implemented on low-cost IoT sensor nodes or have limited applicability, which motivates our design of sparse Gaussian matrix in the next section. 

% Traditional projection matrices, which can be used in the projection operation, include dense random matrices (e.g., random Gaussian or Bernoulli matrices), Toeplitz matrices, and also sparse binary measurement matrices. In this paper, however, we employ the proposed sparse Gaussian matrix as an alternative. The advantages of sparse Gaussian matrix compared with traditional matrices are validated through experiments in the next section. 

% As shown in Fig.~\ref{fig:csLink}, the basic idea of this lossy transmission process is that, instead of directly transmitting the originally sampled signal, we perform a dimension-reduced projection on it before transmission using a well-designed projection matrix. Then, the generated set of sensor measurements are transmitted to the FC, which may undergo data loss. According to CS theory, we can disregard the lost sensor measurements, and reconstruct the original signal based on the correctly received data using CS techniques. Note, however, that in order to achieve successful signal recovery, the number of received sensor measurements should exceed a threshold determined by CS theory. 

% As stated in~\cite{wu2015efficient}, the data loss during transmission is modeled as a RCS process, where the corresponding projection matrix is a sparse binary measurement matrix which contains one and only one ``$1$" element in each row and at most one ``$1$" in each column, and ``$0$" everywhere else. Formally, it is constructed as follows

\section{Sparse Gaussian Matrix Design}
\label{sec3}
In this section, we present the construction method of the sparse Gaussian matrix, and theoretically prove its RIP.

% denoted by ${\bf{\Phi}}_{sG}$, is designed, which can be easily and efficiently implemented on low-cost IoT sensor nodes due to its sparse structure. The RIP of sparse Gaussian matrix is proven theoretically, which guarantees accurate signal recovery with high probability. 

\subsection{Sparse Gaussian Matrix Construction}
Given a series of numbers $\{ {x_i}\} _{i = 1}^n$ which are independently and identically distributed according to a zero-mean Gaussian distribution, sparse Gaussian matrix is constructed by randomly placing these elements in a null matrix (Note that the size of the null matrix is larger than the number of used elements), i.e.,
\begin{equation}
{\bf{\Phi}}_{sG} = \left[ {\begin{array}{*{20}{c}}
{{x_1}}&0&{..}&{{x_2}}&0\\
0&{..}&{{x_i}}&0&{..}\\
0&{{x_{i + 1}}}&0&{..}&{{x_{i + 2}}}\\
0&{..}&0&{{x_n}}&0
\end{array}} \right],
\end{equation}
or we can construct it in a structured manner as an analogue to Toeplitz matrix, i.e., 
\begin{equation}
{\bf{\Phi}}_{sG} = \left[ {\begin{array}{*{20}{c}}
{{x_1}}&0&{..}&{{x_2}}&0\\
0&{{x_i}}&0&{..}&{{x_{i + 1}}}\\
{{x_{n - 2}}}&0&{{x_{n - 3}}}&{..}&0\\
0&{{x_n}}&{..}&{{x_{n - 1}}}&0
\end{array}} \right].
\end{equation}

When implementing a sparse Gaussian matrix in an IoT sensor node, the first step is to generate a null matrix according to the predefined matrix size. Then, the locations of a small number of nonzero elements are determined according to a random number generator. Finally, the values of the nonzero elements are determined according to a Gaussian-distributed (zero mean and unit variance) random number generator, while the rest locations in the null matrix are set to $0$. 

% In the next subsection, we elaborate on the proof of RIP for sparse Gaussian matrix. 

\subsection{RIP Proof for Sparse Gaussian Matrix}
A Matrix ${\bf{\Phi }}$ is said to satisfy RIP of order $K$ with parameter ${\delta _K} \in (0,1)$, if
\begin{equation}
(1 - {\delta _K})\|{\bf{\alpha }}\|_2^2 \le \|{\bf{\Phi \alpha }}\|_2^2 \le (1 + {\delta _K})\|{\bf{\alpha }}\|_2^2
\end{equation}
holds for all ${\bf{\alpha}} \in {\mathbb{R}}^N$ with no more than $K$ nonzero entries. In other words, $\bf{\Phi}$ satisfies RIP ($K,{\delta _K}$). 

In practical application scenarios, however, the signals of interest (i.e., $\bf{x}$) are usually sparse under an orthogonal transform basis $\bf{\Psi}$ rather than the canonical basis. Therefore, it's more convenient to consider the $\bf{\Psi}$-RIP as follows
\begin{equation}
%\begin{aligned}
(1 - {\delta _K})\|{\bf{x}}\|_2^2 \le \|{\bf{\Phi x}}\|_2^2 \le (1 + {\delta _K})\|{\bf{x}}\|_2^2,
%\end{aligned}
\end{equation}
for all $\bf{x}$ satisfying $\|{{{\bf{\Psi}}^H}{\bf{x}}}\|_0 \le K$, where ${\bf{\Psi}}^H$ denotes the conjugate transpose of $\bf{\Psi}$. 

% ${\bf{\Phi }}$ satisfies RIP if the singular value of all submatrices of ${\bf{\Phi }}$ formed by retaining no more than $S$ columns  of ${\bf{\Phi }}$ are in the range $(\sqrt {1 - {\delta _K}} ,\sqrt {1 + {\delta _K}} )$~\cite{Toeplitz}.

According to the definition of RIP, as long as the eigenvalues of all submatrices' gram matrix are in the range of $( {1 - {\delta _K}} , {1 + {\delta _K}} )$, where submatrices are formed by retaining no more than $K$ columns of ${\bf{\Phi }}$, then the matrix ${\bf{\Phi }}$ satisfies the RIP condition. 

\begin{lemma}\label{lem:Gersgorin}
The eigenvalues of a matrix ${\bf{\Phi }} \in {\mathbb{R}}^{m \times m}$ all lie in the union of $m$ discs ${d_i} = {d_i}({c_i},{r_i})$, $i = 1,2,...,m$, centered at ${c_i} = {{\bf{\Phi }}_{ii}}$ with radius
\begin{equation}
{r_i} = \sum\limits_{j = 1,j \ne i}^m {|{{\Phi} _{ij}}|},
\end{equation}
\end{lemma}
where ${{\Phi }_{ii}}$ and ${{\Phi }_{ij}}$ denote the diagonal and off-diagonal elements of $\bf{\Phi}$, respectively. For the proof of Lemma~\ref{lem:Gersgorin}, please refer to~\cite{varga2010gervsgorin}. 

To prove that the eigenvalues of all submatrices' gram matrix are in the range of $( {1 - {\delta _K}} , {1 + {\delta _K}} )$, we resort to the following Lemma. 
\begin{lemma}\label{lem2}
Suppose ${\delta _d},~{\delta _o} > 0$ and ${\delta _d} + {\delta _o} = {\delta _K} \in (0,1)$, and if ${G_{ii}}$, the diagonal elements of matrix ${\bf{G}} = {\bf{\Phi }}^T{\bf{\Phi }}$ satisfy $|{G_{ii}} - 1| < {\delta _d}$ and the off-diagonal elements ${G_{ij}~(i \ne j)}$ satisfy $|{G_{ij}}| < {\delta _o}/K$, then ${\bf{\Phi }}$ satisfies RIP ($K,{\delta _K}$). 
\end{lemma}
\begin{proof}
Please see Appendix A in the online appendix~\cite{Online_appendix}. 
\end{proof}

If the elements of $\bf{\Phi}$ are independently and identically distributed according to a Gaussian distribution, then ${G_{ii}}$ is the sum of squares of Gaussian random variables, and ${G_{ij}~(i \ne j)}$ is the sum of products between independent Gaussian random variables. The range of ${G_{ii}}$ and ${G_{ij}~(i \ne j)}$ can be determined by the following Lemma. 
\begin{lemma}\label{lem3}
Given $\{ {x_i}\} _{i = 1}^n$, assume that it has $kn$ ($0<k<1$) nonzero elements obeying independent and identical Gaussian distribution with zero mean and variance $1/kn$, then the sum of squares of ${x_i}$'s satisfies 
\begin{equation}
\Pr \left( {|\sum\limits_{i = 1}^n {x_i^2 - 1} | \ge {\delta _d}} \right) \le 2\exp ( - \frac{{nk\delta _d^2}}{{16}}).
\end{equation}
\end{lemma}
\begin{proof}
Please see Appendix B in the online appendix~\cite{Online_appendix}.  
\end{proof}

\begin{lemma}\label{lem4}
Given two sets of numbers $\{ {x_i}\} _{i = 1}^n$ and $\{ {y_i}\} _{i = 1}^n$, each with $kn$ nonzero elements, which are independently and identically distributed according to a Gaussian distribution with zero mean and variance $1/kn$, then
\begin{equation}
\Pr \left( {|\sum\limits_{i = 1}^n {{x_i}{y_i}} | \ge t} \right) \le 2\exp ( - \frac{{n{t^2}}}{{4 + 2t}}),
\end{equation}
where $0 \le t \le 1$.
\end{lemma}

\begin{proof}
Assume $\{ {x_i}\} _{i = 1}^n$ and $\{ {y_i}\} _{i = 1}^n$ are independent and identical Gaussian random variables with zero mean and variance $1/n$, and based on~\cite{haupt2010toeplitz}, we have 
\begin{equation}\label{equ:Toeplitz2}
\Pr \left( {|\sum\limits_{i = 1}^n {{x_i}{y_i}} | \ge t} \right) \le 2\exp ( - \frac{{{t^2}}}{{4{\sigma ^2}(n{\sigma ^2} + t/2)}}).
\end{equation}
Substituting $n=kn$ and ${\sigma ^2} = 1/n$ into Eq.~(\ref{equ:Toeplitz2}), we have Lemma~\ref{lem4}. 
\end{proof}

\begin{theorem}
\label{RIPSG}
Assume that ${\bf{\Phi}} \in {\mathbb{R}}^{m \times n}(n > 2)$ is a sparse matrix, whose nonzero elements are independently and identically distributed according to a Gaussian distribution with zero mean and and variance $1/m$, then for any ${\delta _K} \in (0,1)$, there exist constants ${c_1}$ and ${c_2}$ which are only dependent on ${\delta _K}$ and $K$, such that as long as $m > {c_2}\ln n$, $\bf{\Phi}$ satisfies RIP($K,{\delta _K}$) with probability exceeding $1 - \exp ( - {c_1}m)$. 
\end{theorem}

\begin{proof}
Please see Appendix C in the online appendix~\cite{Online_appendix}.  
\end{proof}

Theorem~\ref{RIPSG} has proven the RIP of sparse Gaussian matrix ${\bf{\Phi}}_{sG}$. In our application scenario of compressive data collection in IoT, however, the equivalent projection matrix is ${{\bf{\Phi}}_r}{\bf{\Phi}}_{sG}$, where ${\bf{\Phi}}_r$ is a sparse binary measurement matrix. In other words, ${{\bf{\Phi}}_r}{\bf{\Phi}}_{sG}$ is a submatrix of ${\bf{\Phi}}_{sG}$ formed by randomly selecting a subset of rows from ${\bf{\Phi}}_{sG}$. Therefore, we next present the RIP of ${{\bf{\Phi}}_r}{\bf{\Phi}}_{sG}$. 

\begin{corollary}\label{colla}[RIP for submatrices (by selecting a subset of rows) of sparse Gaussian matrix]
Assume that ${\bf{\Phi}} \in {\mathbb{R}}^{m \times n}(n > 2)$ is a sparse matrix, whose nonzero elements are independently and identically distributed according to a Gaussian distribution with zero mean and variance $1/n$, then with regard to submatrices of $\bf{\Phi}$ by selecting $m'$ rows from $\bf{\Phi}$, there exist constants ${c_3}$ and ${c_4}$ which are only dependent on ${\delta _K}$, $K$ and the proportion $\mu  = \frac{{m'}}{m}$, such that as long as $m > {c_4}\ln n$, submatrices of $\bf{\Phi}$ satisfy RIP($K,{\delta _K}$) with probability exceeding $1 - \exp ( - {c_3}m)$. 
\end{corollary}
\begin{proof}
Please see Appendix D in the online appendix~\cite{Online_appendix}. 
\end{proof}

\section{Performance Evaluation}
\label{sec6}
% In this section, we evaluate the performance of employing the designed sparse Gaussian matrix to perform CS-based source coding for compressive data collection in IoT, and make a comparison with traditional projection matrices. We utilize a real-world dataset from SensorScope~\cite{ingelrest2010sensorscope} and tailor a temperature signal of size $N = 1024$ as the original signal to be transmitted in the experiments. We would like to point out that we focus mainly on single-hop transmission based data collection in IoT in this paper, and thus the process of data collection (as detailed in Section~\ref{sec5}) from different sensor nodes is identical. Therefore, in the subsequent experiments, we only consider the case where one sensor node transmits data to the FC over a lossy channel, and the results reported in this paper can be easily extended to the case where there are multiple sensor nodes transmitting data to the FC. 
In this section, we evaluate the performance of employing the designed sparse Gaussian matrix to perform CS-based source coding for compressive data collection in IoT and make a comparison with traditional projection matrices. We utilize a real-world dataset from SensorScope~\cite{ingelrest2010sensorscope} and tailor a temperature signal of size $N = 1024$ as the original signal to be transmitted in the experiments. We would like to point out that we focus mainly on single-hop transmission-based data collection in IoT in this paper. Thus, the process of data collection (as detailed in Section~\ref{sec5}) from different sensor nodes is identical. Therefore, in the subsequent experiments, we only consider the case where one sensor node transmits data to the FC over a lossy channel, and the results reported in this paper can be easily extended to the case where there are multiple sensor nodes transmitting data to the FC.

\begin{figure}
\centering
\begin{minipage}[t]{0.24\textwidth}
  \centering
\includegraphics[width=\textwidth]{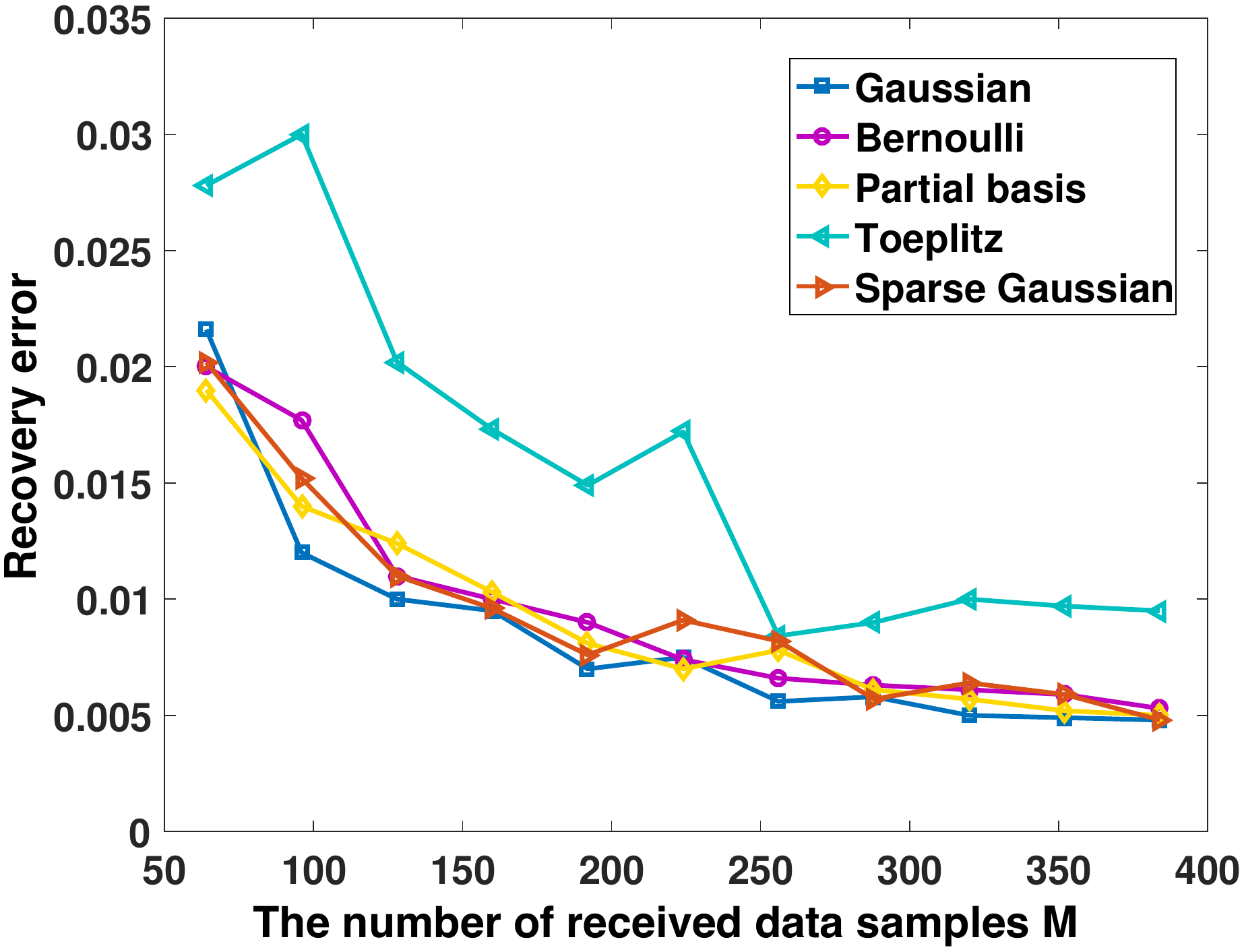}
\caption{Signal recovery performance comparison of using different projection matrices.}
\label{fig:recoverM}
\end{minipage}
\begin{minipage}[t]{0.235\textwidth}
  \centering
  \includegraphics[width=\textwidth]{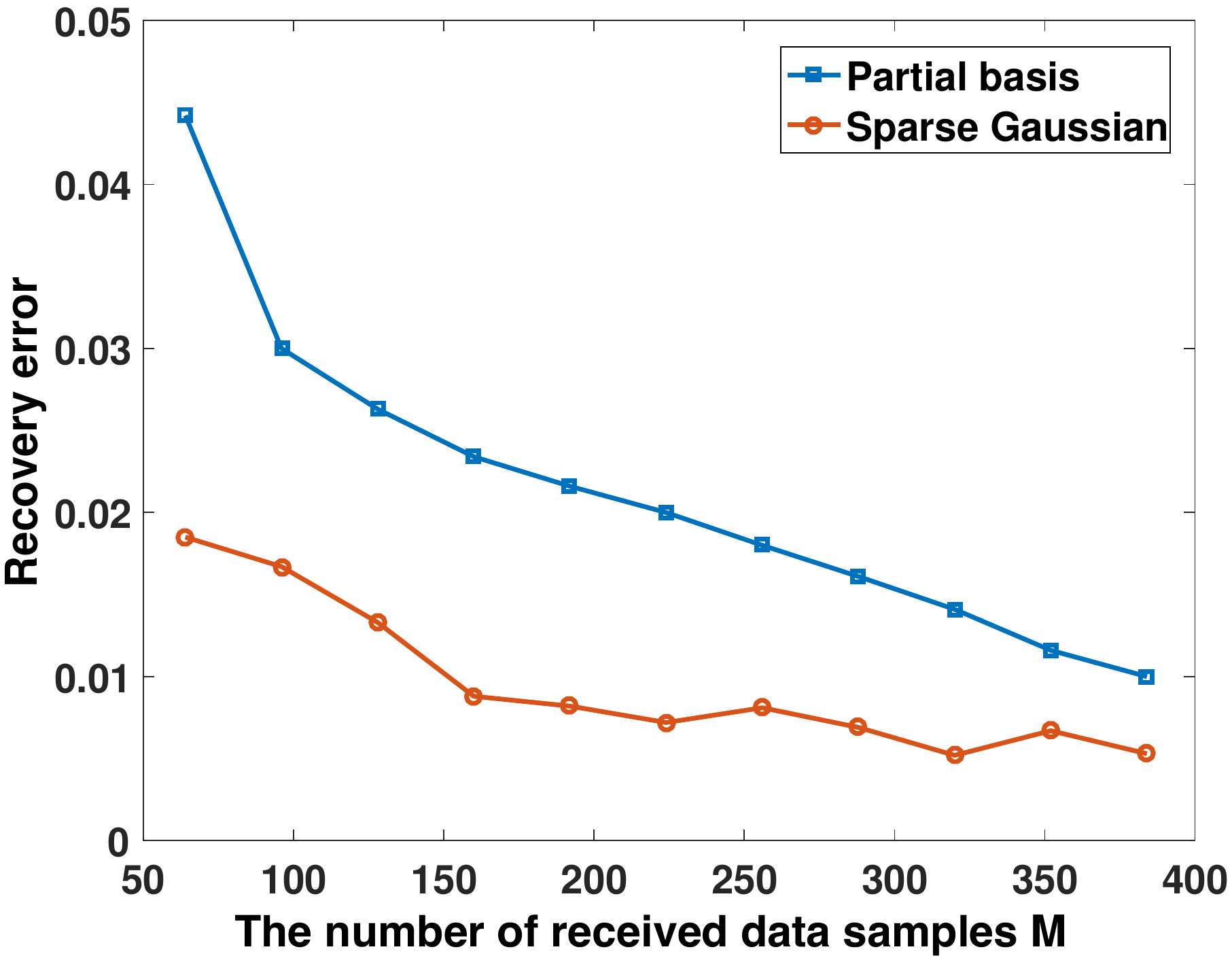}
\caption{Signal recovery performance comparison with packet length set to $16$.}
\label{fig:recoverML}
\end{minipage}
\end{figure}

\subsection{Signal Recovery Performance}
% We first make a comparison between sparse Gaussian matrix and traditional projection matrices in terms of signal recovery performance. The result of signal recovery error comparison among using different projection matrices at the sensor node before transmission is shown in Fig.~\ref{fig:recoverM}. Note that in this experiment, we assume that the projection operation is dimensional-invariant and the number of received data samples $M$ varies from $64$ to $384$ with a step length of $32$. Herein, the used sparse Gaussian matrix is of size $1024 \times 1024$, and it is constructed as follows. First, we generate a null matrix of the same size. Then, we randomly choose $10\%$ (rounded to the nearest integer) locations in the null matrix to place nonzero elements. Finally, the values of the nonzero elements are determined according to a Gaussian-distributed (zero mean and unit variance) random number generator, while the rest locations in the null matrix are set to $0$. Note that in the subsequent experiments, the used sparse Gaussian matrices are all constructed in the same way unless otherwise specified. We can observe from the figure that utilizing dense random matrices (i.e., Gaussian and Bernoulli matrices), partial basis matrices and the proposed sparse Gaussian matrix achieves similar signal recovery performance. However, when using Toeplitz matrices, a larger recovery error is obtained under the same number of received data samples. This is because the structure of Toeplitz matrices is broken during lossy transmission from the sensor node to the FC. 
We first compare sparse Gaussian matrix and traditional projection matrices in terms of signal recovery performance. The result of signal recovery error comparison among using different projection matrices at the sensor node before transmission is shown in Fig.~\ref{fig:recoverM}. Note that in this experiment, we assume that the projection operation is dimensional-invariant, and the number of received data samples $M$ varies from $64$ to $384$ with a step length of $32$. Herein, the used sparse Gaussian matrix is of size $1024 \times 1024$, and it is constructed as follows. First, we generate a null matrix of the same size. Then, we randomly choose $10\%$ (rounded to the nearest integer) locations in the null matrix to place nonzero elements. Finally, the values of the nonzero elements are determined according to a Gaussian-distributed (zero mean and unit variance) random number generator, while the rest locations in the null matrix are set to $0$. Note that in the subsequent experiments, the used sparse Gaussian matrices are all constructed in the same way unless otherwise specified. We can observe from the figure that utilizing dense random matrices (i.e., Gaussian and Bernoulli matrices), partial basis matrices, and the proposed sparse Gaussian matrix achieves similar signal recovery performance. However, when using Toeplitz matrices, a more significant recovery error is obtained under the same number of received data samples. This is because the structure of Toeplitz matrices is broken during lossy transmission from the sensor node to the FC.

% Considering that Toeplitz matrices achieve poor signal recovery performance, and that dense random matrices induce high resource cost for sensor nodes (which will be validated later), we next make a further comparison between sparse Gaussian matrix and partial basis matrix with packet length taken into consideration. Assume the packet length is set to $16$, that is, $16$ consecutive data samples are first organized into a data packet, and then all the generated data packets are transmitted to the FC. The signal recovery error comparison result is shown in Fig.~\ref{fig:recoverML}. As noted in the figure, when packet length is set to $16$, using partial basis matrix as the projection matrix results in larger recovery error compared with using sparse Gaussian matrix. This is because when adopting a large packet length, if a data packet is lost during transmission, all the data samples in it will be missing, leading to block data loss. However, partial basis matrices are only suitable to the case of random data loss (i.e., the packet length is $1$, that is, each transmitted data sample is organized into an independent data packet). 
Considering that Toeplitz matrices achieve poor signal recovery performance and that dense random matrices induce high resource cost for sensor nodes (which will be validated later), we next make a further comparison between sparse Gaussian matrix and partial basis matrix with packet length taken into consideration. Assume the packet length is set to $16$, that is, $16$ consecutive data samples are first organized into a data packet, and then all the generated data packets are transmitted to the FC. The result of signal recovery error comparison is shown in Fig.~\ref{fig:recoverML}. As noted in the figure, when packet length is set to $16$, using partial basis matrices as the projection matrix results in more significant recovery error than using sparse Gaussian matrix. This is because when adopting a large packet length, if a data packet is lost during transmission, all the data samples in it will be missing, leading to block data loss. However, partial basis matrices are only suitable to the case of random data loss (i.e., the packet length is $1$, that is, each transmitted data sample is organized into an independent data packet).

\begin{figure}
\centering
\begin{center}
\subfigure[Partial basis matrix]{
\begin{minipage}[t]{0.45\columnwidth}
\includegraphics[width=1\textwidth]{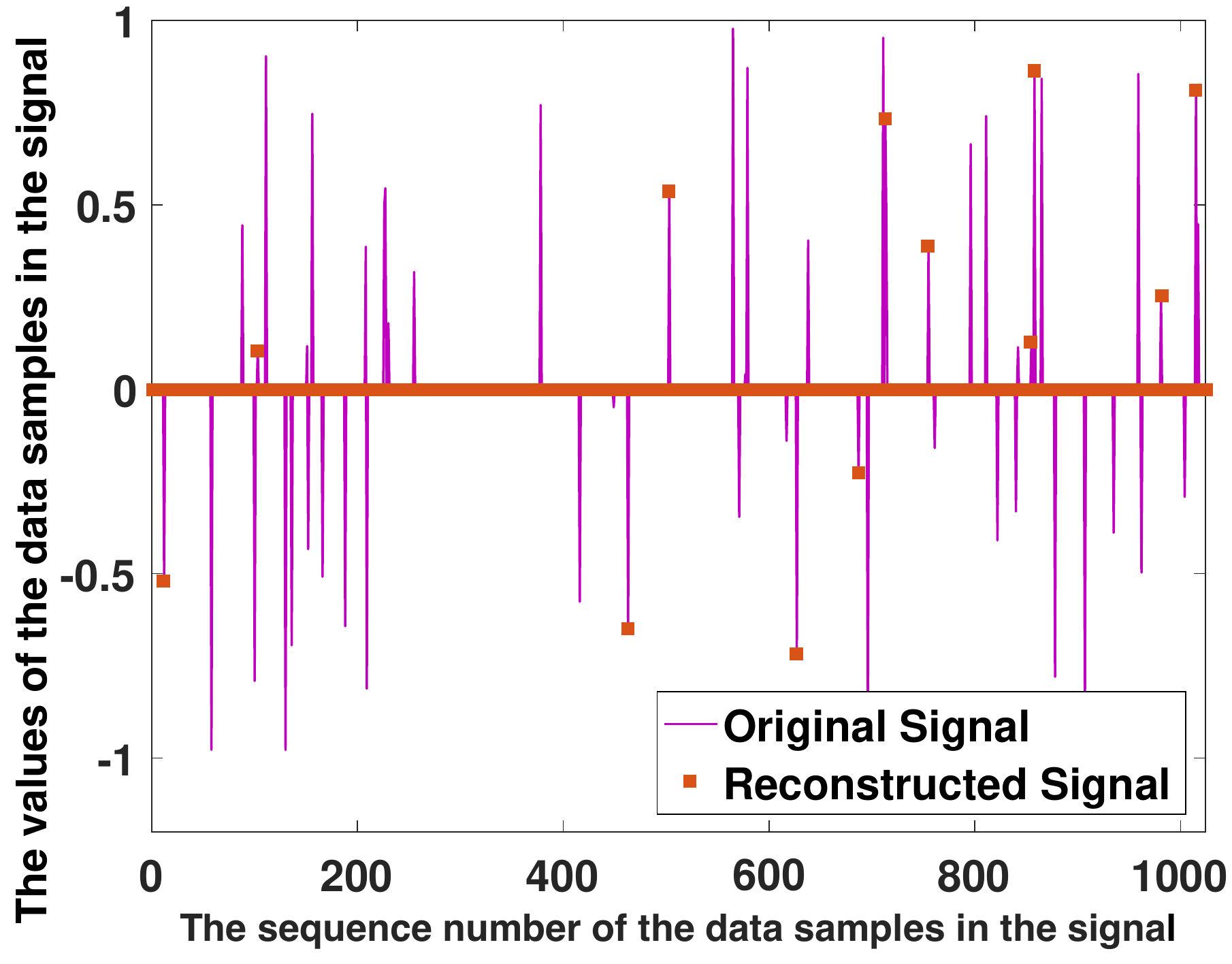}
\label{fig:sparseTrcs}
\end{minipage}
}
\subfigure[Sparse Gaussian matrix]{
\begin{minipage}[t]{0.45\columnwidth}
\includegraphics[width=1\textwidth]{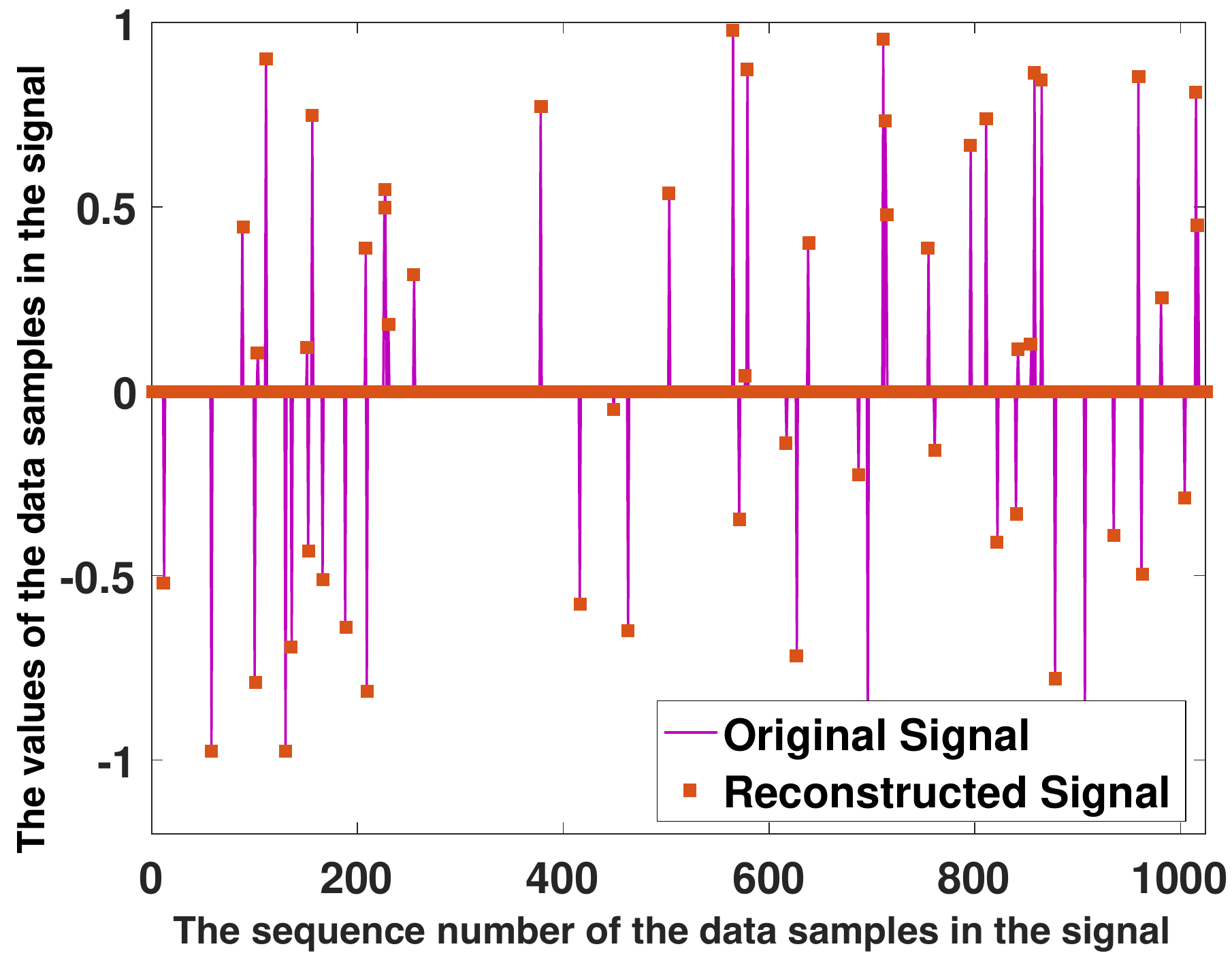}
\label{fig:sparseTsg}
\end{minipage}
}

\caption{Signal recovery comparison using a sparse signal in the canonical basis.}
\vspace{-15pt}
\label{fig:sparseT}
\end{center}
\end{figure}

% Another limitation of partial basis matrices is that they can only be applied to the scenario when the target signal is sparse under some orthogonal transform basis. However, when the target signal is itself sparse, that is, the sparsifying basis is the canonical basis, partial basis matrices fail to achieve successful signal recovery. In order to demonstrate that partial basis matrices cannot promise accurate signal recovery when the target signal is sparse under the canonical basis. We synthesize a sparse signal $\bf{x}$ under the canonical basis with a dimension of $N = 1024$. Only $50$ elements of $\bf{x}$ are nonzero (varying from $-1$ to $1$) while most elements are zero. We respectively employ partial basis matrices and the designed sparse Gaussian matrix to perform CS-based source coding at the IoT sensor node, and the number of received data samples is set to $64$. The signal recovery results using the two kinds of projection matrix are shown in Fig.~\ref{fig:sparseT}. We can clearly observe from the figure that accurate signal recovery is achieved when using the designed sparse Gaussian matrix to perform CS-based source coding, while it fails to achieve accurate signal recovery when using partial basis matrices, which further validates the superior performance of the designed sparse Gaussian matrix and the limitations of partial basis matrices. 
Another limitation of partial basis matrices is that they can only be applied to the scenario when the target signal is sparse under some orthogonal transform basis. However, when the target signal is itself sparse. That is, the sparsifying basis is the canonical basis, partial basis matrices fail to achieve successful signal recovery. To demonstrate that partial basis matrices cannot promise accurate signal recovery when the target signal is sparse under the canonical basis. We synthesize a sparse signal $\bf{x}$ under the canonical basis with a dimension of $N = 1024$. Only $50$ elements of $\bf{x}$ are nonzero (varying from $-1$ to $1$) while most elements are zero. We respectively employ partial basis matrices and the designed sparse Gaussian matrix to perform CS-based source coding at the IoT sensor node, and the number of received data samples is set to $64$. The signal recovery results using the two kinds of projection matrix are shown in Fig.~\ref{fig:sparseT}. We can observe from the figure that accurate signal recovery is achieved when using the designed sparse Gaussian matrix to perform CS-based source coding, while it fails to achieve accurate signal recovery when using partial basis matrices, which further validates the superior performance of the designed sparse Gaussian matrix and the limitations of partial basis matrices.

\subsection{Resource Cost}
% \begin{figure}
% \centering
% \includegraphics[width=0.5\columnwidth]{pdf//STM32.pdf}
% \caption{Sensor node with STM32W108 chip.}
% \label{fig:STM32}
% \end{figure}

% \begin{table}
% \centering
% \caption{\label{STMparameter}STM32W108 Parameters}
% \begin{tabular}{l|l}
% \hline
% Clock speed &24MHz \\
% RAM &12KB\\
% ROM &128KB\\
% Power &4 AA batteries\\
% \hline
% \end{tabular}
% \end{table}

Considering that many low-cost IoT sensor nodes are resource-constrained (e.g., computational capability, storage capacity), it is inefficient or even impractical to implement complicated operations on them. Now, we make a comparison between sparse Gaussian matrix and dense random matrices (e.g., Gaussian and Bernoulli matrices) in terms of time and memory cost for projection matrix construction. We implement the three kinds of projection matrices on STM32W108 chips that offer IEEE $802.15.4$ communication standard, with a clock speed of $24$ Mhz, a RAM of $12$ KB, a ROM of $128$ KB, and $4$ AA batteries for power supply.  

% The IoT sensor node with STM32W108 chip is shown in Fig.~\ref{fig:STM32}, and the relevant parameters (e.g., memory, clock speed) are shown in Table~\ref{STMparameter}. 

\addtolength{\topmargin}{0.011in}
% \begin{figure}[t]
% \centering
% \begin{center}
% \subfigure[Time cost for matrix construction]{
% \begin{minipage}[t]{0.45\columnwidth}
% \includegraphics[width=1\textwidth]{pdf//timecost_new.pdf}
% \label{fig:sgt}
% \end{minipage}
% }
% \subfigure[Memory cost for matrix construction]{
% \begin{minipage}[t]{0.45\columnwidth}
% \includegraphics[width=1\textwidth]{pdf//memory_new.pdf}
% \label{fig:sgm}
% \end{minipage}
% }

% \caption{Resource cost comparison among Gaussian, Bernoulli and sparse Gaussian matrix. }
% \label{fig:subTsg}
% % \vspace{-5pt}
% \end{center}
% \end{figure}

\begin{figure}
\centering
\begin{minipage}[t]{0.24\textwidth}
  \centering
\includegraphics[width=\textwidth]{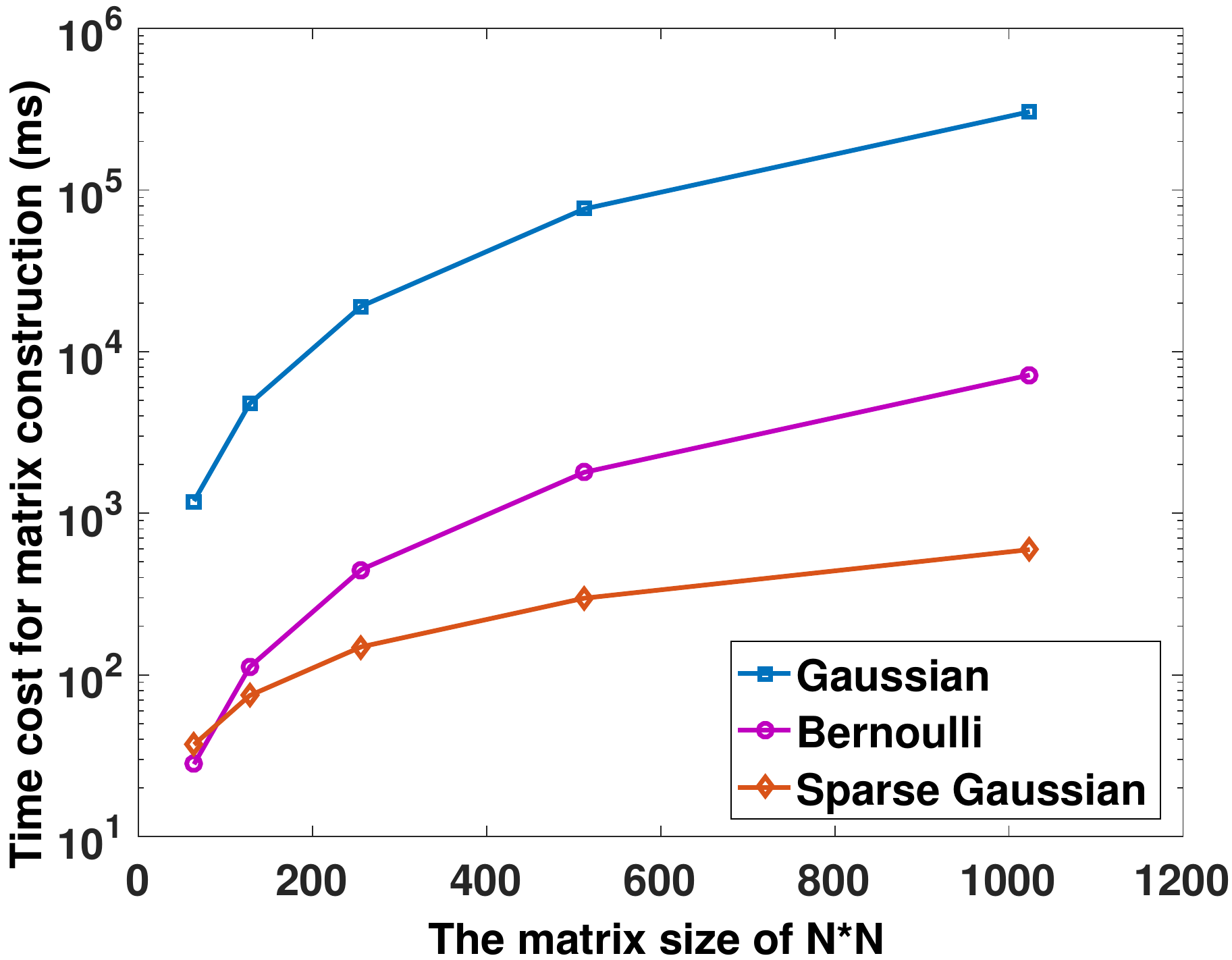}
\caption{Time cost comparison among Gaussian, Bernoulli, and sparse Gaussian matrix for matrix construction.}
\label{fig:sgt}
\end{minipage}
\begin{minipage}[t]{0.24\textwidth}
  \centering
  \includegraphics[width=\textwidth]{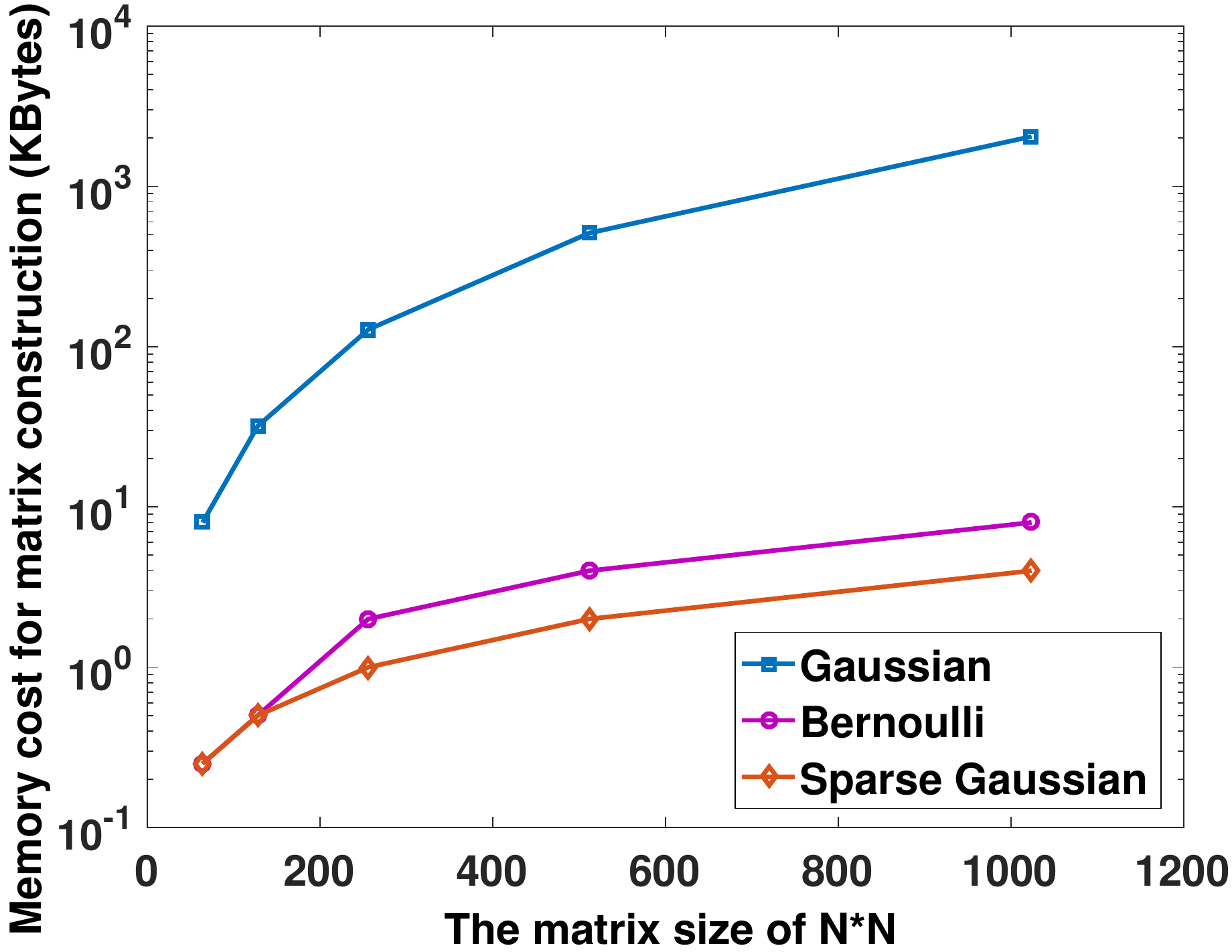}
\caption{Memory cost comparison among Gaussian, Bernoulli, and sparse Gaussian matrix for matrix construction.}
\label{fig:sgm}
\end{minipage}
\end{figure}
% The comparison results of resource cost are shown in Fig.~\ref{fig:subTsg} (the matrix is of size $N \times N$). Note that a random Gaussian number is quantified by $16$ bits and one element in Bernoulli matrix is quantified by $1$ bit (note that when there are only a few nonzero elements in Bernoulli matrix, we only store the positions of these nonzero elements, which is more efficient). Also, it is obvious that the number of nonzero elements in sparse Gaussian matrix will directly impact the resource cost of sensor nodes. We can see from Fig.~\ref{fig:sgt} that it costs the most time to construct a random Gaussian matrix, which is self-explanatory. When $N = 64$, sparse Gaussian matrix costs a little bit more time than Bernoulli matrix, and remains the most efficient one to construct when $N$ increases. As for the memory cost shown in Fig.~\ref{fig:sgm}, sparse Gaussian matrix consumes almost the same memory as Bernoulli matrix when $N = 64$, while remains the most memory-saving for rest cases. Specifically, when $N =1024$, constructing a Bernoulli matrix with only two nonzero elements in each row needs $8$ KBytes of memory, while only $4$ KBytes are enough for sparse Gaussian matrix. Besides, it costs $1792$ and $763$ ms to construct such a Bernoulli matrix and sparse Gaussian matrix, respectively. 
The comparison results of time cost and memory cost for matrix construction are respectively shown in Fig.~\ref{fig:sgt} and Fig.~\ref{fig:sgm} (the matrix is of size $N \times N$). Note that a random Gaussian number is quantified by $16$ bits, and one element in Bernoulli matrix is quantified by $1$ bit. Note that when there are only a few nonzero elements in Bernoulli matrix, we only store the positions of these nonzero elements, which is more efficient. Also, the number of nonzero elements in sparse Gaussian matrix will directly impact the resource cost of sensor nodes. We can see from Fig.~\ref{fig:sgt} that it costs the most time to construct a random Gaussian matrix, which is self-explanatory. When $N = 64$, sparse Gaussian matrix costs a little more time than Bernoulli matrix and remains the most efficient one to construct when $N$ increases. As for the memory cost shown in Fig.~\ref{fig:sgm}, sparse Gaussian matrix consumes almost the same memory as Bernoulli matrix when $N = 64$, while remains the most memory-saving for rest cases. Specifically, when $N =1024$, constructing a Bernoulli matrix with only two nonzero elements in each row needs $8$ KBytes of memory, while only $4$ KBytes are enough for sparse Gaussian matrix. Besides, it costs $1792$ and $763$ ms to construct such a Bernoulli matrix and sparse Gaussian matrix, respectively.

\section{Conclusion}

\label{sec7}
%Although CS has been applied in random channel access,
% In this paper, in order to achieve efficient and resource-saving CS-based source coding for compressive data collection in IoT, we design a novel simple projection matrix, named sparse Gaussian matrix, which is easy and efficient to be implemented in practical IoT systems. Specifically, we first present the construction method of sparse Gaussian matrix, and then prove its RIP theoretically. Finally, extensive experiments are conducted to demonstrate that employing the designed sparse Gaussian matrix to perform CS-based source coding at the sensor node can significantly save time and memory cost while guaranteeing satisfactory signal recovery performance. 
In this paper, to achieve efficient and resource-saving CS-based source coding for compressive data collection in IoT, we design a novel, simple projection matrix, named sparse Gaussian matrix, which is easy and efficient to implement practical IoT systems. Specifically, we first present the construction method of sparse Gaussian matrix and then prove its RIP theoretically. Finally, we conduct extensive experiments to demonstrate that employing the designed sparse Gaussian matrix to perform CS-based source coding at the sensor node can significantly save time and memory cost while guaranteeing satisfactory signal recovery performance.

\ifCLASSOPTIONcaptionsoff
  \newpage
\fi

\bibliographystyle{IEEEtran}
    \bibliography{reference}

\end{document}